\documentclass[11pt,3p,review,authoryear]{article}

\usepackage[english]{babel}
\usepackage[utf8x]{inputenc}
\usepackage{dsfont}
\usepackage{amsmath}
\usepackage{amsthm}
\usepackage{amssymb}
\usepackage{natbib}
\usepackage{graphicx}
\usepackage{subfigure}
\usepackage{indentfirst}

\usepackage[colorlinks]{hyperref}

\usepackage{anysize}
\marginsize{3cm}{3cm}{2.5cm}{2.5cm}

\newtheorem{prop}{Proposition}
\newcommand{\D}{\mathrm{\bf d}}
\newcommand{\PE}{\mathbb{E}}
\newcommand{\Var}{\textrm{Var}}

\newcommand{\mbf}[1]{\mathbf{#1}}

\newcommand{\comb}[2]{\left(\begin{array}{c}#1\\ #2\end{array}\right)}

\begin{document}
\title{Modeling and forecasting daily average PM$_{10}$ concentrations by a seasonal ARFIMA model with volatility}
\author{Valdério Anselmo Reisen$^{12}$, Alessandro José Q. Sarnaglia$^1$, Neyval C. Reis Jr.$^2$,\\ Céline Lévy-Leduc$^3$ and Jane Méri Santos$^2$}
\date{{\footnotesize $^1$Statistics Department, Federal University of Espírito Santo\\ $^2$Graduate Program in Environmental Engineering, Federal University of Espírito Santo\\ $^3$Département TSI, Télécom ParisTech}}
\maketitle

\begin{abstract}
This paper considers the possibility that the daily average Particulate Matter ($PM_{10}$) concentration is a seasonal fractionally integrated process with time-dependent variance (volatility). In this context, one convenient extension is to consider the SARFIMA model \citep{reisen:rodrigues:palma:2006a,reisen:rodrigues:palma:2006b} with GARCH type innovations. The model is theoretically justified and its usefulness is corroborated with the application to $PM_{10}$ concentration in the city of Cariacica-ES (Brazil). The model adjusted was able to capture the dynamics in the series. The out-of-sample forecast intervals were improved by considering heteroscedastic errors and they were able to identify the periods of more volatility.\\

\end{abstract}

\section{Introduction}

The issue of airborne ambient Particulate Matter ($PM$) has become a well-recognized research topic in environmental sciences. Epidemiological studies have reported strong associations between $PM_{10}$ concentrations ($PM$ with an aerodynamic diameter of less than or equal to 10 $\mu m$) and several adverse health effects, including respiratory problems in children, death and increased hospital admissions for cardiopulmonary and respiratory conditions see, for example, \citet{touloumi:etal:2004}, \citet{perez:etal:2007}, \citet{zelm:etal:2008} and references therein.

In the literature, several modeling strategies have been developed or optimized for the study and forecast of $PM$ concentration in urban areas, such as \citet{diazrobles:etal:2008}, \citet{konovalov:etal:2009} and others. Among these modeling efforts, statistical models based on multiple regression \citep{stadlober:hormann:pfeiler:2008} and time series tools, such as the Box-Jenkins time series Autoregressive Integrated Moving Average (ARIMA) model, have been widely used for this class of problems \citep{goyal:chan:jaiswal:2006,liu:2009}.

Models which adequately describe the physical behavior of the data are essential for accurately forecasting in any area of application. In this paper, a Seasonal Autoregressive Fractionally Integrated Moving Average (SARFIMA) model with more than one fractional parameter and a non-constant conditional error variance (heteroscedastic errors) is used to illustrate how it can be useful to fit and forecast series with seasonality, volatility and long-range dependency (or long-memory) features. These time series phenomena are quite common characteristics found in data in many areas of interest. For example, \citet{windsor:toumi:2001} analyzed the variability of the pollutants ozone and PM with long-memory technique which was also the methodology applied by \citet{baillie:chung:tieslau:1996} to model and forecast temperature series. \citet{karlaftis:vlahogianni:2009} studied the memory and volatility properties in transportation time series \citet{kumar:ridder:2010} focused on modeling and forecasting ozone episodes through heteroscedastic processes (GARCH) associated with ARIMA model.

Roughly speaking, seasonality is a phenomenon where the observation in the instant, say $t$, is highly correlated with the one in the time $t-s$. In this case, $s$ is called season length. It is important to consider statistical tools which take into account the seasonality effect. However, some studies focusing on the forecast of daily $PM_{10}$ concentrations do not regard for the seasonal influence of weather patterns \citep{goyal:chan:jaiswal:2006}. Other studies, such as \citet{stadlober:hormann:pfeiler:2008} try to control the seasonal component by using \textit{dummy} variables which is suitable just in the case when seasonality is present in the mean structure only.

Time series with volatility is characterized by a non-constant conditional variance, i.e., the error variance changes as a function of time. This fact contrasts with the usual assumption, namely the variance of the process is assumed to be constant. However, if the variance is time-varying, the forecast variance can be reduced by accommodating the conditional variance which will lead to more accurate forecast confidence intervals. A systematic structure for modeling volatility in a time series is the Autoregressive Conditional Heteroscedastic (ARCH) model proposed by \citet{engle:1982}. An extension of this model, the Generalized Autoregressive Conditional Heteroscedastic (GARCH), was proposed by \citet{bollerslev:1986}. See also \citet{bollerslev:chou:kroner:1992} for a more complete review on this subject. Due to the high temporal variability of $PM_{10}$ concentration, it is usually found to have a time-varying conditional variance (see \citet{chelani:devotta:2005} among others). Volatility models are popular tools in financial literature, however, only recently, these have caught the attention of many researchers interested in modeling time-varying variance in time series of environmental sciences, e. g. \citet{mcaleer:chan:2006}.

Recently, time series analysis with long-term dependency have been studied by several authors in different areas of applications. In the time-domain, long-range dependency is usually characterized by a significant autocorrelation even for those observations separated by a relatively long time period. The ARFIMA model \citep{granger:joyeux:1980,hosking:1981} is a time series model that well accommodates the long-memory feature. As discussed in the next section, this model has the parameter $d$, which governs the memory of the process. Several estimation methods for the long-memory parameter have been proposed. The most popular semiparametric estimator is due to \citet{geweke:porterhudak:1983}, \citet{reisen:1994} among others. The usefulness of modeling time series with the long-memory characteristic by ARFIMA processes has been extensively studied, theoretically and empirically, in many areas, such as mathematics, economics among others. For a recent review of this subject, see \citet{palma.2007}. The characteristics of the long-memory parameter estimators have been extensively investigated under various model situations, such as the presence of non-Gaussian errors and outliers, e.g. \citet{sena:reisen:lopes:2006}, \citet{fajardo:reisen:cribarineto:2009} among others.

However, in environmental science, more specifically, in the air pollution area, the use of the ARFIMA model has still not been well explored. Nowadays, there is a lot of software that makes using this model less difficult in applied works. So, due to the important model features of the ARFIMA process, this model is certain to motivate much research in the near future in the environmental science area. \citet{iglesias:jorqueira:palma:2006} is an example of applied work with long-memory process in the air pollution area. The authors   have investigated the use of an ARFIMA model to handle time series of $PM_{2.5}$, $PM_{10}$ concentrations and other gaseous pollutants.

A natural extension of the ARFIMA model to accommodate seasonal features is the seasonal ARFIMA model. Since the early 90's, this model has caught the attention of researchers that are interested in studying long-memory time series with seasonal fractional parameters. \citet{porterhudak:1990} among others proposed the use of \citet{geweke:porterhudak:1983} method for the estimation of seasonal ARFIMA processes. A generalization of these seasonal long-memory models are the ARUMA and the GARMA models, which were originally proposed by \citet{giraitis:leipus:1995} and \citet{woodward:cheng:gray:1998}, respectively. \citet{reisen:rodrigues:palma:2006a,reisen:rodrigues:palma:2006b} presented studies regarding the seasonal ARFIMA model, which is a particular case of the ARUMA/GARMA models, and suggested long-memory estimators. Empirical studies, performed by the authors, indicate the efficiency of the estimators when compared to other existing methods. Seasonality and long-memory properties have been explored theoretically and empirically by a large number of works, see for example, \citet{reisen:rodrigues:palma:2006a,reisen:rodrigues:palma:2006b}, \citet{arteche:robinson:2000} among others.

For a series that presents seasonal long-memory features with conditional variance (or volatility), one convenient extension is to consider the SARFIMA model with GARCH type innovations. This model can provide a useful way of analyzing a process exhibiting seasonal long-memory with volatility. This is the main purpose of this paper, which proposes the use of a SARFIMA model with one non-seasonal and one seasonal fractional parameter and GARCH errors. The model is theoretically justified and its usefulness is corroborated with the application to $PM_{10}$ ambient concentrations.

The rest of this paper is organized as follows. Section 2 introduces the model and discusses its properties. The section also summarizes the estimation method of the  parameters. Section 3 deals with the analysis and modeling of the $PM_{10}$ contaminant and forecasting issues. Some conclusions are draft in Section 4.

\section{The model and parameter estimation}\label{models}
A process $X_t\equiv\{X_t\}_{t\in\mathbb{Z}}$ is defined as a zero-mean  SARFIMA$(p,d,q)\times(P, D, Q)_s$ model with non-seasonal orders $p$ and $q$, seasonal orders $P$ and $Q$, difference parameters $d$ and $D$, and season length $s\in \mathds{N^*}= \mathds{N}-\{0\}$  if
\begin{equation}\label{def:U_t}
 U_t= \nabla^{\D} X_t
\end{equation} is a SARMA $(p,q)\times(P,Q)_s$ process. That is, the process $\{U_t\}_{t\in\mathbb{Z}}$ satisfies \begin{equation}\label{def:SARMA} \Phi(B^s)\phi(B) U_t=\Theta(B^s)\theta(B)\epsilon_t\; , \end{equation} where $\{\epsilon_t\}_{t\in\mathbb{Z}}$ is a white noise with $\PE(\epsilon_t)=0$ and $\Var(\epsilon_t)=\sigma_\epsilon^2$ and $B$ is the backward operator satisfying $B Y_t=Y_{t-1}$ for any process $\{Y_t\}_{t\in\mathbb{Z}}$.

In (\ref{def:U_t}), the operator $\nabla^{\D}$ is defined by:
\begin{equation}\label{def:nabla}
 \nabla^{\D}= (1-B)^{d}(1-B^s)^{D}\; ,
 \end{equation}
 where $\D=(d,D)\in\mathbb{R}^2$ is the memory vector parameter, $d$ and $D$ are the fractionally parameters at the zero (or long-run) and seasonal frequencies, respectively. Also, the fractional filters are $$ (1-B^k)^x=\sum_{j=0}^\infty\comb{x}{j}\left(-B^k\right)^j, k=1,s, \hbox{ and } x=d,D, $$ where $$
\comb{x}{j}=\frac{\Gamma(x+1)}{\Gamma(j+1)\Gamma(x-j+1)},
$$
and $\Gamma(\cdot)$ is the well-known gamma function.

In (\ref{def:SARMA}), the polynomials $\Phi(\cdot)$, $\Theta(\cdot)$, $\phi(\cdot)$ and $\theta(\cdot)$ are given by
\begin{align*}
\Phi(z^s)&=1-\Phi_1z^s-\Phi_2z^{2s}-\cdots-\Phi_Pz^{Ps}\;,\\
\Theta(z^s)&= 1-\Theta_1z^s-\Theta_2z^{2s}-\cdots-\Theta_Qz^{Qs}\;,\\
\phi(z)&=1-\phi_1z-\phi_2z^2-\cdots-\phi_pz^p\;,\\
\theta(z)&= 1-\theta_1z-\theta_2z^2-\cdots-\theta_qz^q\;.
\end{align*}
It is assumed that these polynomials have no common zeros and satisfy the conditions $\Phi(z^s)\phi(z)\neq0$ and $\Theta(z^s)\theta(z)\neq0$ for $|z|=1$. Futhermore, in the above equations, $(\Phi_i)_{1\leq i\leq P}$, $(\Theta_j)_{1\leq j\leq Q}$, $(\phi_k)_{1\leq k\leq p}$ and $(\theta_\ell)_{1\leq \ell\leq q}$ are unknown parameters. For more details, see, for example, \citet{palma:chan:2005}, \citet{giraitis:leipus:1995} among others. If $|d+D|<1/2$ and $|D|<1/2$, $X_t$ is a stationary and invertible process and, at seasonal frequency $\omega_s \in [-\pi,\pi]$, the spectral density becomes unbounded and behaves as
\begin{equation}
 f(\omega + \omega_s) \sim C {\left|s\omega\right|}^{-2D} {\left|2\sin \frac{\omega_s}{2} \right|}^ {-2d} \quad \ \omega \rightarrow 0,
\end{equation}
where $C$ is a non-negative constant.

\citet{granger:joyeux:1980} and \citet{hosking:1981} proposed an ARFIMA$(p,d,q)$ model, which is a particular case of the SARFIMA model (Eq. (\ref{def:U_t}) and (\ref{def:SARMA})) when $P=Q=D=0$. The ARFIMA models are commonly used to model time series with long-memory behavior and have the following characteristics; the ARFIMA process is stationary and invertible, when $|d|< 0.5$; $d>0$ characterizes a long-memory dependence; $d=0$ and $d<0$ indicate that the process has a short and an intermediate dependence, respectively. The spectral density function of the ARFIMA model has the form $f(w)\sim C|w|^{-2d}$ for $w\rightarrow0$, where $C$ is a non-negative constant. The correlation between $X_t$ and $X_{t+k}$ satisfies $\rho(k)\sim k^{2d-1}$ as $k\rightarrow\infty$. To estimate $d$, in the context of semiparametric frameworks, the method proposed by \citet{geweke:porterhudak:1983} (GPH) was the pioneering one and it has been widely used in the literature. Based on the GPH method, other variant estimators for $d$ were proposed, for example, \citet{reisen:1994}, \citet{arteche:robinson:2000} and \citet{Reisen:moulines:Soulier:Franco:2010}. Here, the GPH method is the basis of the fractional seasonal and non-seasonal parameter estimation tool.

Let $\{X_1,\ldots,\ X_n\}$ be a sample from the process $X_t$ (Eq. (\ref{def:U_t})). \citet{reisen:rodrigues:palma:2006a,reisen:rodrigues:palma:2006b} suggested a slight modification of \citet{geweke:porterhudak:1983} method to estimate the parameters $d$ and $D$, in a seasonal ARFIMA process (Eq. (\ref{def:U_t})). For a set of Fourier frequencies $\omega_j = \frac{2\pi j}{n}, 1\leq j \leq M = [\frac{(n-1)} {2}]$, where $\lfloor x\rfloor$ is the greatest integer small than or equal to $x$, the estimation method consists in obtaining the estimator $\hat{\D}=(\hat{d},\hat{D})$ from the approximated multiple linear regression equation
\begin{equation}
\log I(\omega_j)\cong a_0-D\log\left[2\sin\left(\frac{s\omega_j}{2}\right)\right]^2-d\log\left[2\sin\left(\frac{\omega_j}{2}\right)\right]^2+u_j,\label{regre}
\end{equation}
where the periodogram function $I(\omega_j)$ is given by \begin{equation*} 	I(w_j)=\frac{1}{2\pi n}\left|\sum_{t=1}^{n}X_te^{iw_j t}\right|^2, \end{equation*} $a_0$ is a constant and $$ u_j=\log\frac{I(w_j)}{f_X(w_j)}-\mathbb{E}\left[\log\frac{I(w_j)}{f_X(w_j)}\right]. $$

Under some model conditions, \citet{reisen:etal:2010} establish that
\begin{equation}
\sqrt{M}(\hat{\D}-\D)\rightarrow\mathcal{N}\left(W^{-1}b,\frac{\pi^2}{6}W^{-1}\right)\label{dist}\;
\end{equation}
where $b$ and $W$ are a vector and a matrix $2\times2$ of constants, respectively, and $M$  is the bandwidth in Equation \ref{regre} that satisfies
$$
\left(\frac{M}{n}\right)^{\iota}\log M+\frac{1}{M}\rightarrow 0\; ,\; \mbox{ as }\; n\rightarrow \infty\; , \label{bandwth}
$$
for some $\iota>0$.

The high variability of the data suggests that the $PM_{10}$ has a time varying conditional variance \citep{chelani:devotta:2005}. Thus, it may be interesting and useful to model $PM_{10}$ with a statistical tool that incorporates the features seasonality, long-memory and heteroscedasticity. Thus, the SARFIMA process defined in Eq. (\ref{def:U_t}) and (\ref{def:SARMA}), with heteroscedastic errors, is  the model candidate to adjust and forecast daily average concentrations of $PM_{10}$.

Due to the extensive literature on the application of ARFIMA and GARCH to model time series with long-memory and heteroscedasticity features, the ARFIMA process with GARCH innovations becomes a very popular tool in practical data analysis. This model was the main motivation of the work \citet{ling:li:1997}. The authors introduced the ARFIMA$(p,d,q)$-GARCH$(r,m)$ model, where $p,\ q,\ r,\ m \in \mathds{N^*}$ and $d\in\mathds{R}$, and presented model and maximum likelihood estimator properties. Independently, \citet{sena:reisen:lopes:2006} investigated empirically the ARFIMA$(p,d,q)$-GARCH$(r,m)$ model with parametric and semiparametric estimation procedures to estimate the parameters of the ARFIMA part. \citet{baillie:chung:tieslau:1996} analyzed inflation series from 10 countries with ARFIMA-GARCH methodology. They suggested a procedure to obtain approximate maximum likelihood estimates of an ARFIMA-GARCH model. These works give strong support to use the ARFIMA model in a practical application even in the case where the errors have heteroscedastic properties. Then, based on this discussion, the seasonal model defined in Eq. (\ref{def:U_t}) and (\ref{def:SARMA}) can be extended to a seasonal model with heteroscedastic errors such as GARCH$(r,m)$ process. The model that incorporates these characteristics is defined hereafter as SARFIMA$(p,d,q)\times(P, D, Q)_s$-GARCH$(r,m)$, where now $\{\epsilon_t\}$ in Eq. (\ref{def:SARMA}) has the following structure
\begin{equation}
 \epsilon_t|\Im_{t-1} \sim D(0,h_t), \ \ h_t=\alpha_0+\sum_{i=1}^{m}\alpha_i\epsilon_{t-i}^2+\sum_{j=1}^{r}\beta_jh_{t-j},\label{garch}
\end{equation}
 where $m,r\in \mathbb{N^*}$ represent the model orders, $\alpha_0>0$ and $\alpha_i,\beta_j\geq0$, for $ i=1,2,...,m$ and $j=1,2,...,r$, and $\Im_t$ denotes the $\sigma$ field generated by the past information $\{\epsilon_{t-1},\epsilon_{t-2},\cdots\}$. In above, $D$ is a probability distribution of a continuous random variable, for example, normal or t-student distribution.

Combining the model properties in \citet{reisen:rodrigues:palma:2006a,reisen:rodrigues:palma:2006b} with Therorem 2.3 given in \citet{ling:li:1997}, the following proposition is established for the SARFIMA$(p,d,q)\times(P, D, Q)_s$-GARCH$(r,m)$ model.

\begin{prop}\label{proparfimagarch}
Let $X_t$ be generated by Eq. (\ref{def:U_t}) and (\ref{def:SARMA}) with $\epsilon_t$ given by (\ref{garch}) where $\sum_{i=1}^{m}\alpha_i+\sum_{j=1}^{r}\beta_j < 1$. Suppose that the polynomials $\Phi(z^s)\phi(z)$ and $(z^s)\theta(z)$ in (\ref{def:SARMA}) have no common zeros and that $\D$ in (\ref{def:nabla}) satisfies: $|d+D|<1/2$ and $|D|<1/2$. Then, the following statements hold
\begin{itemize}
\item[(a)] If $\Phi(z^s)\phi(z)\neq0$, for $|z|=1$, then $X_t$ is second-order stationary and has the unique representation given by
\begin{equation}\label{mainf}
\end{equation}
where $\psi_j$ are determined by the Laurent expansion
$$
\sum_{j=0}^{\infty}\psi_jz^j=\frac{\Theta(z^s)\theta(z)}{\Phi(z^s)\phi(z)}\;,
$$
in some annulus of $|z|=1$. Hence, $X_t$ is strictly stationary and ergodic.
\item[(b)] If $\Theta(z^s)\theta(z)\neq0$, for $|z|\leq1$, then $X_t$ is invertible and
$$
\sum_{j=0}^{\infty}\psi_j^*\frac{\Phi(z^s)\phi(z)}{\Theta(z^s)\theta(z)}X_{t-j}=\epsilon_{t}\;,
$$
where $\psi_j^*$ are given by
\begin{equation}\label{phistar}		
\psi_j^*=\pi_j+\sum_{i=1}^{\infty}\pi_i^{(s)}\pi_{j-is},
\end{equation}
with
$$
\pi_l = \frac{\Gamma (l-d)}{\Gamma (l+1) \Gamma (-d)}, \quad l=0,1,\ldots\;,
$$
$$
\pi_k^{(s)} = \frac{\Gamma (k-D)}{\Gamma (k+1) \Gamma (-D)}, \quad k=0,1,\ldots\; ,
$$
where $ \Gamma (.)$ is the Gamma function.
\item[(c)] The spectral density of $\{X_t\}$ is given by
\begin{equation}\label{spec}
f_X(\omega)=f_U(w)\left[2\sin\left(\frac{s\omega}{2}\right)\right]^{-2D}\left[2\sin\left(\frac{\omega}{2}\right)\right]^{-2d},\; \omega\in [-\pi,\pi]\;,
\end{equation}
where $f_U(\cdot)$ is the spectral density of the stationary SARMA process $\{U_t\}$ and $U_t=\nabla^{\D}X_t$.
\end{itemize}
\end{prop}
The proof of this proposition is given in the Appendix.

Next section presents the analysis of daily average $PM_{10}$ concentrations based on the SARFIMA$(p,d,q)\times(P, D, Q)_s$-GARCH$(r,m)$ model previously introduced.

\section{Analysis and results of modeling PM$_{10}$ concentration}

As previously mentioned, the daily average PM$_{10}$ concentration is the data set here analyzed to illustrate the methodology previously discussed. The series is expressed in $\mu$g/m$^3$ and it was observed in Cariacica, which belongs to the Metropolitan Region of Greater Vitória (RGV)-ES- Brazil. RGV is comprised of five cities with a population of approximately 1.7 million inhabitants in an area of 1,437 $km^2$. The region is situated in the South Atlantic coast of Brazil (latitude 20\textdegree 19S, longitude 40\textdegree 20W) and has a tropical humid climate, with average temperatures ranging between 24\textdegree C and 30\textdegree C. 

The raw series has a sample size of 1826 observations, measured from January 1st of 2005 to December 31st of 2009.  The series has mean $\bar{X} =43.81 \mu g /m^3$ and it is shown graphically in Figure \ref{y}. Maximum concentration is generally observed in the winter months from July to September and the data shows to be stationary in a mean-level with strong seasonality pattern as expected. In addition, there is considerable evidence that the conditional variance is not constant over time, so that conditional volatility models seem to be appropriate choice for capturing the time-varying volatility in the level of the PM$_{10}$ concentration. For modeling purpose, the time series is divided into two parts; learning and prediction sets. The 1603 observations from January 1st of 2005 until May 22nd of 2009 are considered as learning set and the remaining 233 observations are considered for the prediction study (these observations are representing by a dashed line in Figure \ref{y}). 

\begin{figure}[!ht]\centering
  \includegraphics[scale=0.55]{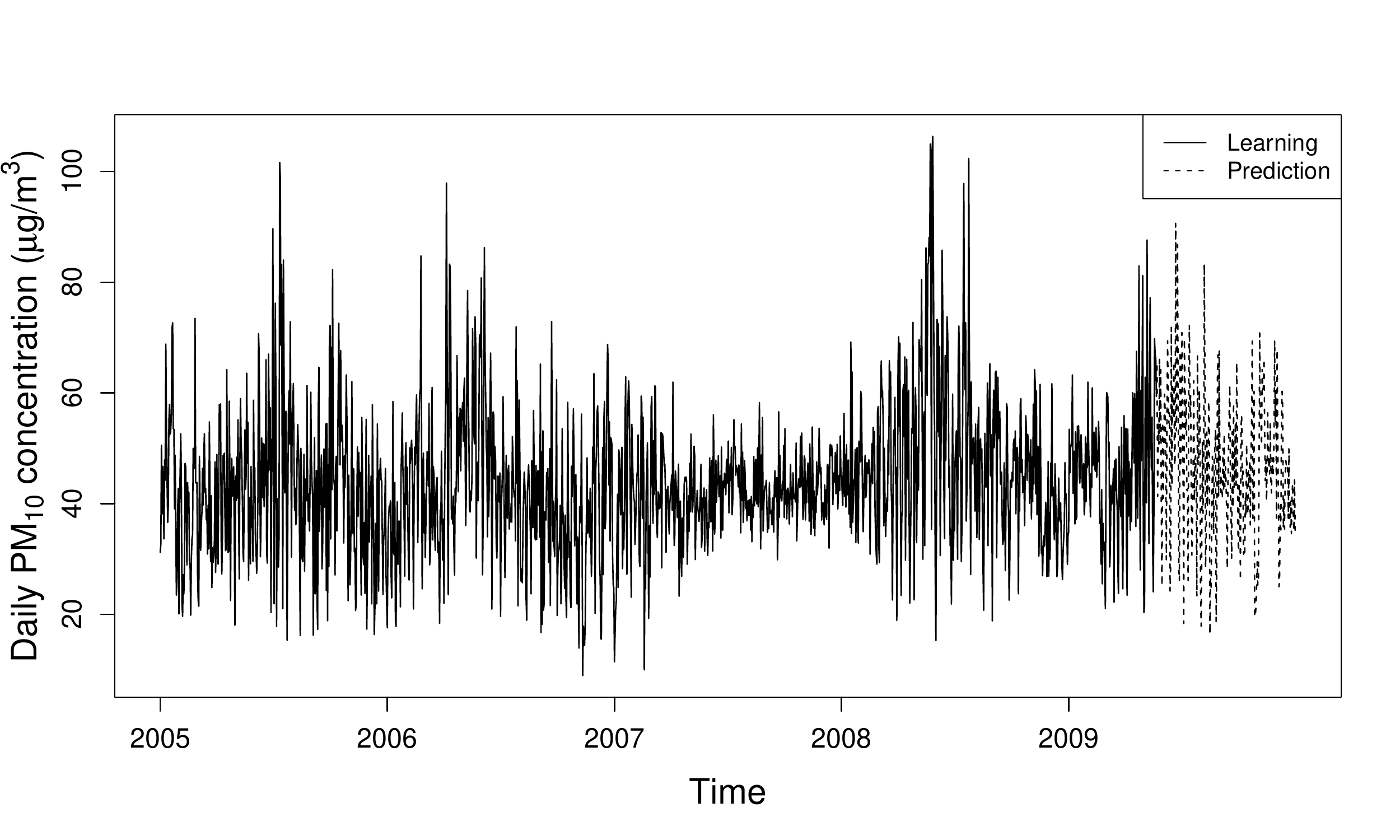}
\caption{Daily PM$_{10}$ concentration in $\mu g/m^3$ from 01/01/2005 to 12/31/2009}\label{y}
\end{figure}

The sample autocorrelation (ACF) and partial autocorrelation (PACF) functions of PM$_{10}$ are shown in Figures \ref{fig_acf_serie} and \ref{fig_pacf_serie}, respectively. These plots clearly show the presence of the seasonality behavior with period $s=7$, which is an expected data behavior since the series was observed daily. The frequency domain counterpart of the sample ACF is the periodogram which is presented in Figure \ref{periodogram}. The sample spectrum has peaks at frequencies very close to zero and also at frequencies which are multiples of 1/7.


\begin{figure}[!ht]
\centering
\subfigure[ACF of PM$_{10}$ concentration]{
  \includegraphics[scale=0.45]{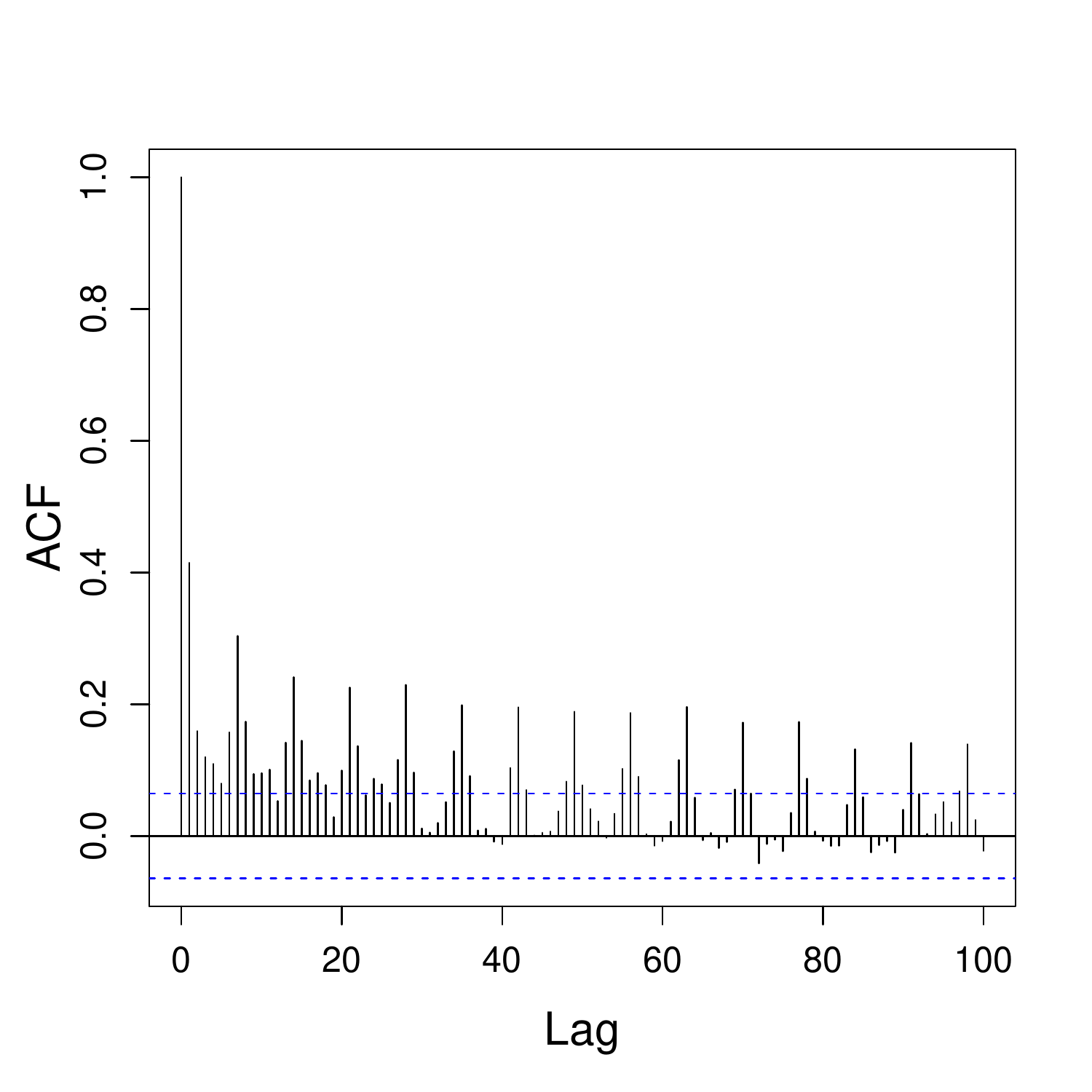}\label{fig_acf_serie}
}
\subfigure[PACF of PM$_{10}$ concentration]{
  \includegraphics[scale=0.45]{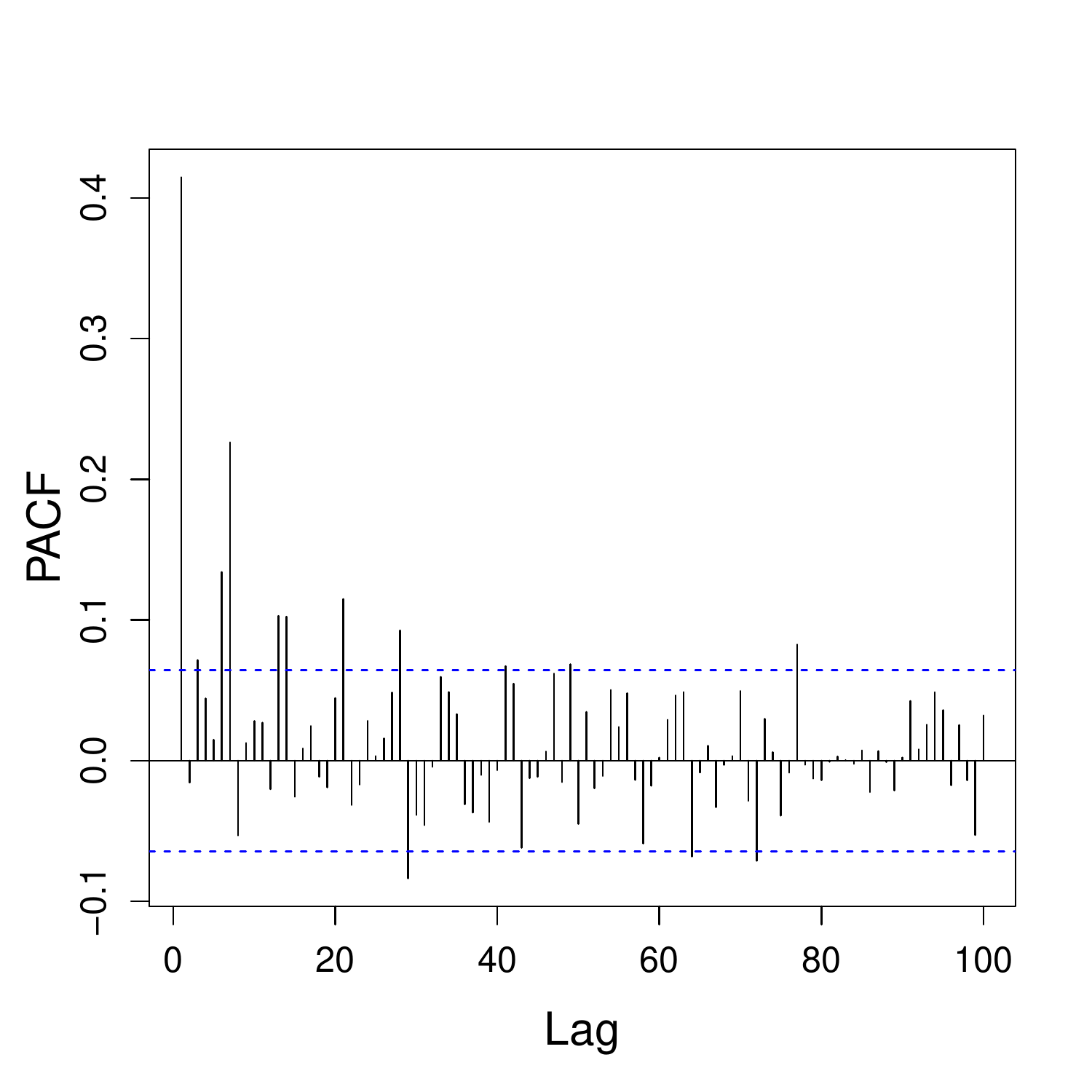}\label{fig_pacf_serie}
}\\
\subfigure[Periodogram of PM$_{10}$ concentration]{
  \includegraphics[scale=0.55]{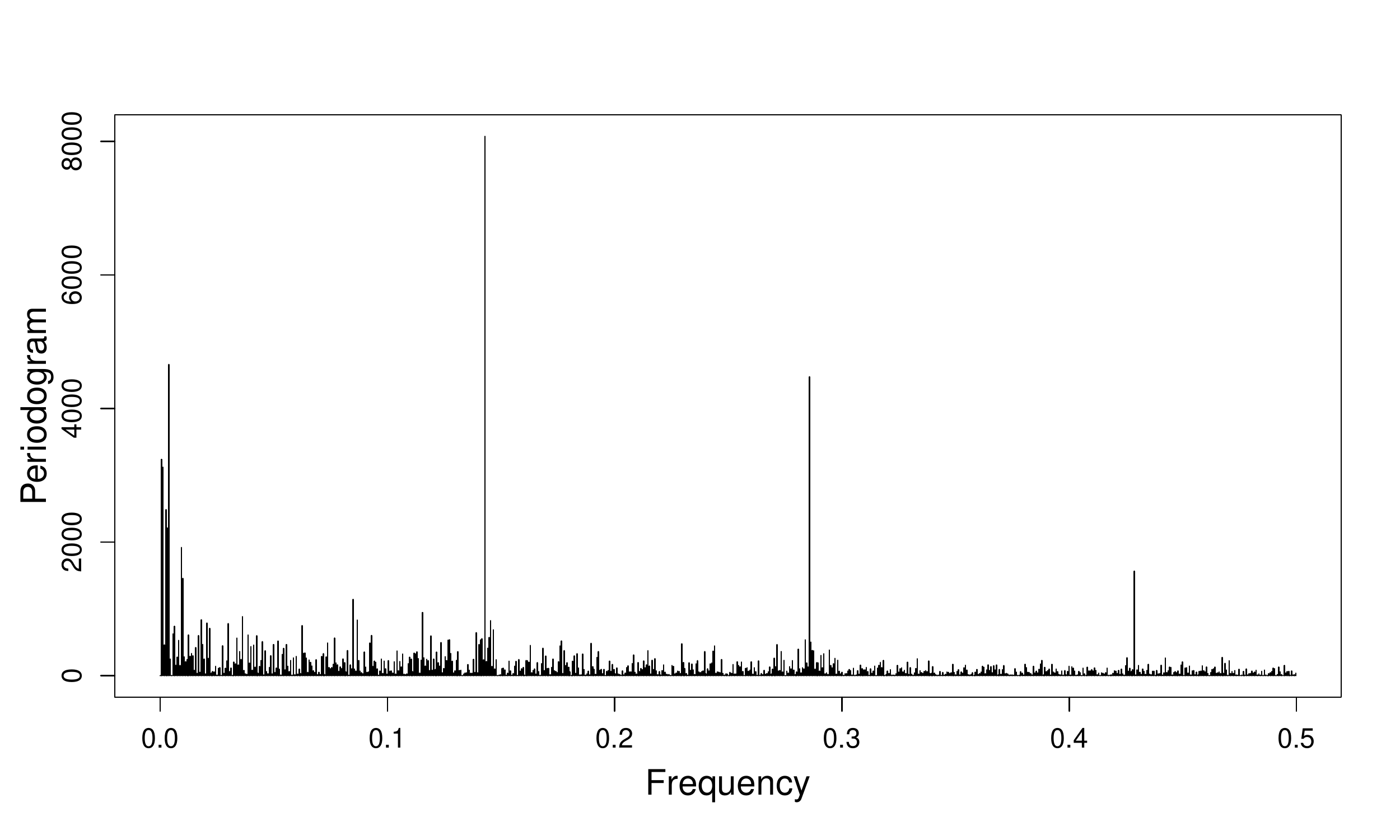}\label{periodogram}
}
\caption{ACF, PACF and periodogram of the PM$_{10}$ dataset.}
\end{figure}

An interesting feature observed from the sample ACF is the positive, significant and slowly decaying of the sample autocorrelations in the first lags, at the lags multiple of 7 and in the lags between the seasonal periods. This may indicates long-memory effect in the data with positive fractional seasonal and non-seasonal parameters. The suspicion of this phenomenon in the data is also observed in the plot of the periodogram in which there are significant peaks at the long-run and at the seasonal periods. These plots corroborate the need for a model which adequately describes the seasonal and nonseasonal long-memory behaviors. However, it is not clear the existence of short-memory parameters by only examining these plots.

The empirical evidence described in above motivates the use of the SARFIMA model defined previously. The SARFIMA modeling strategy follows the same steps suggested in \citet{hosking:1981} and investigated empirically by \citet{reisen:1994} and \citet{reisen:lopes:1999} among others. Firstly, the fractional parameters are estimated by using the  semiparametric tool described in the previous section. This was carried out by using different sizes of the bandwidth $M$. To determine the bandwidth sizes, $M=\lfloor \frac{[(n-s)/2-1]^\alpha}{s}\rfloor$, $0<\alpha<1$. Secondly, the truncated filter $(1-B)^{\hat{d}}(1-B^s)^{\hat{D}}$ is used to filter the observation and to obtain a new series which approximately follows a SMA$(1)\times(1)_7$ model. This new series is used to achieve the complete short-memory model structure. The estimating models and their accuracy are discussed in the next sub-sections. All estimates were computed using $R$ programming language.

\subsection{Adjusted models}

Table \ref{bandwidth} presents the results of the memory estimates obtained from different bandwidths ($M$). The values in brackets correspond to the standard deviations. It can be seen that the estimates of the long-run component described by the fractional differencing parameter $d$ are stables across the bandwidth values. Large  $M$ gives less power for the seasonal frequencies than the smaller ones. The decreasing power of $D$ with $M$ may indicated that there are some contributions of seasonal short-memory counterpart in the model. Since the effect of the seasonal and non-seasonal short-run components can not be avoided in the fractional estimates, the regression equation should be estimated with fewer periodogram ordinates at the zero and at the seasonal frequencies. Thus, the fractional estimates were chosen for $\alpha=0.78$. Note that the stationary model conditions is guaranteed since $0<|\hat{d}+\hat{D}|<0.5$.


\begin{table}[!ht]
\centering
\small
\caption{Estimates of $d$ and $D$ for different \textit{bandwidths} ($M=n^{\alpha}$).}
  \begin{tabular}{ccccccccc}
  \hline\hline
  $\alpha$&&	$M$	 &&	$\hat{d}$ &($sd(\hat{d})$)&&	$\hat{D}$ &($sd(\hat{D})$)	\\
  \cline{1-1}\cline{3-3}\cline{5-6}\cline{8-9}
$0.98$	&&	$99$	&&	$0.2791$	&	$(0.0268)$	&&	$0.1219$	&	$(0.0292)$	\\
$0.96$	&&	$87$	&&	$0.2714$	&	$(0.0276)$	&&	$0.1123$	&	$(0.0307)$	\\
$0.94$	&&	$76$	&&	$0.2623$	&	$(0.0287)$	&&	$0.1157$	&	$(0.0331)$	\\
$0.92$	&&	$66$	&&	$0.2639$	&	$(0.0298)$	&&	$0.1187$	&	$(0.0355)$	\\
$0.90$	&&	$58$	&&	$0.2645$	&	$(0.0310)$	&&	$0.1282$	&	$(0.0383)$	\\
$0.88$	&&	$51$	&&	$0.2496$	&	$(0.0319)$	&&	$0.1423$	&	$(0.0410)$	\\
$0.86$	&&	$44$	&&	$0.2570$	&	$(0.0325)$	&&	$0.1581$	&	$(0.0438)$	\\
$0.84$	&&	$39$	&&	$0.2676$	&	$(0.0331)$	&&	$0.1728$	&	$(0.0463)$	\\
$0.82$	&&	$34$	&&	$0.2707$	&	$(0.0339)$	&&	$0.1704$	&	$(0.0496)$	\\
$0.80$	&&	$29$	&&	$0.2634$	&	$(0.0355)$	&&	$0.1923$	&	$(0.0547)$	\\
$\mbf{0.78}$	&&	$\mbf{26}$	&&	$\mbf{0.2606}$	 &	$\mbf{(0.0372)}$	&&	$\mbf{0.2223}$	&	$\mbf{(0.0596)}$	\\
$0.76$	&&	$22$	&&	$0.2641$	&	$(0.0382)$	&&	$0.2550$	&	$(0.0647)$	\\
  \hline\hline
  \end{tabular}\label{bandwidth}
\end{table}

To obtain the approximation of   $U_t$ (Eq. \ref{def:SARMA}), the observations were filtered by $\nabla^{\hat{\D}}$ truncated at $n=1603$. The new series is $\hat{U}_t=\sum_{j=0}^{n}\hat{\psi}_j^*(X_{t-j}-\bar{X})$, where $\hat{\psi}_j^*$, $j=1,2,\ldots,1603$, are the estimated coefficients $\psi_j^*$ obtained in accordance with (\ref{phistar}) in Proposition \ref{proparfimagarch}. As an example to verify the impact of $X_j$, for large $j$, in the AR infinite representation, the $\hat{\psi}_{1603}^*$ is $\approx$ $10^{-5}$ ($\hat{\psi}_{1603}^*=0.00001340581$), which is nearly zero. Since the observations are in scale of $10^{1}$, the contribution of $X_{j}$ becomes negligible for large $j$.

Figures \ref{fig_acf_serie_filtrada} and \ref{fig_pacf_serie_filtrada} present the sample autocorrelation and partial autocorrelation functions of $\hat{U}_t$, respectively. These plots indicate that a Seasonal Moving-Average (SMA$(1)\times(1)_7$) model may be adequate to describe $\hat{U}_t$. This model order was corroborated by the AIC criterion and residual analysis.

\begin{figure}[!ht]
\centering
\subfigure[ The  ACF of $\hat{U}_t$ ]{
  \includegraphics[scale=0.45]{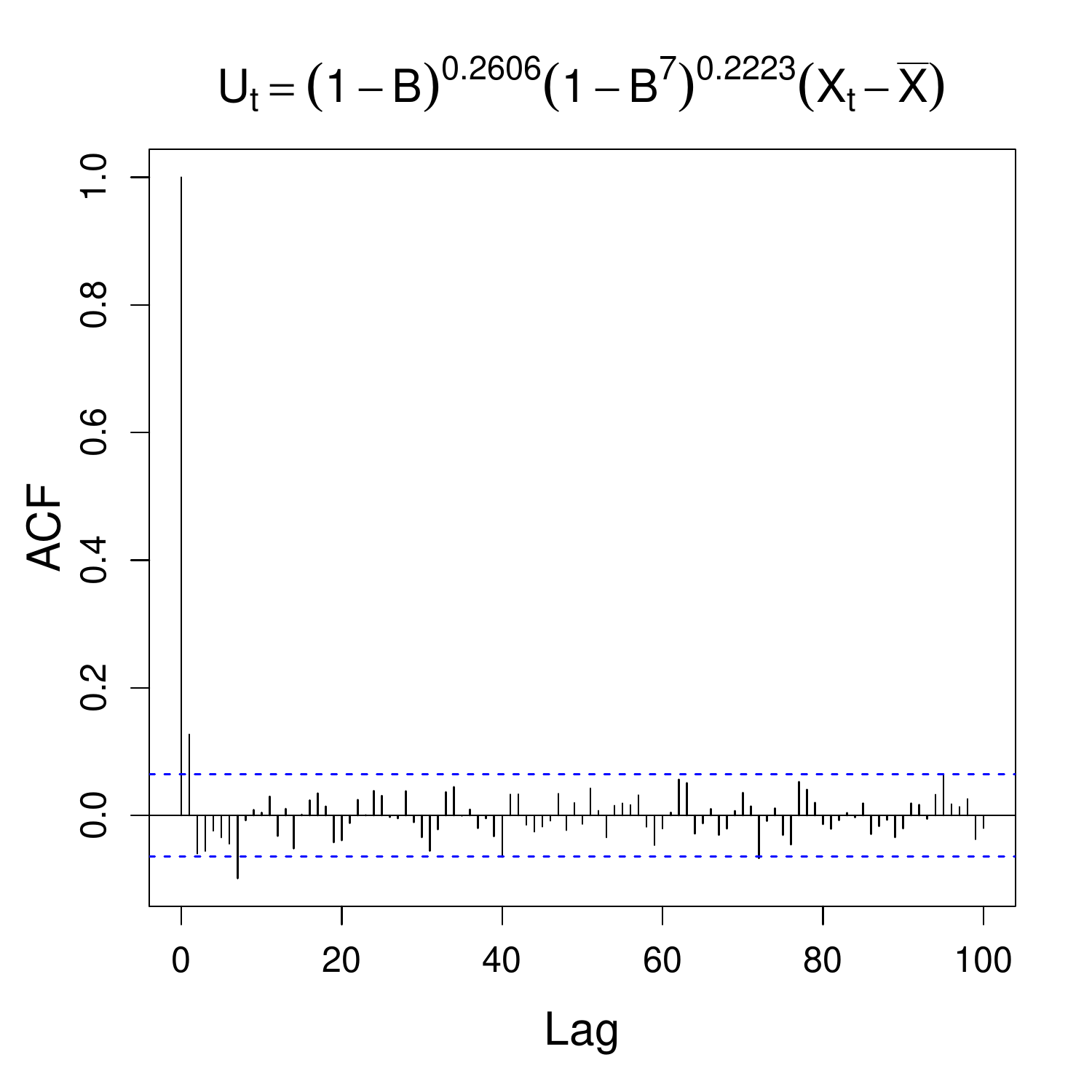}\label{fig_acf_serie_filtrada}
}
\subfigure[The  PACF  of $\hat{U}_t$]{
  \includegraphics[scale=0.45]{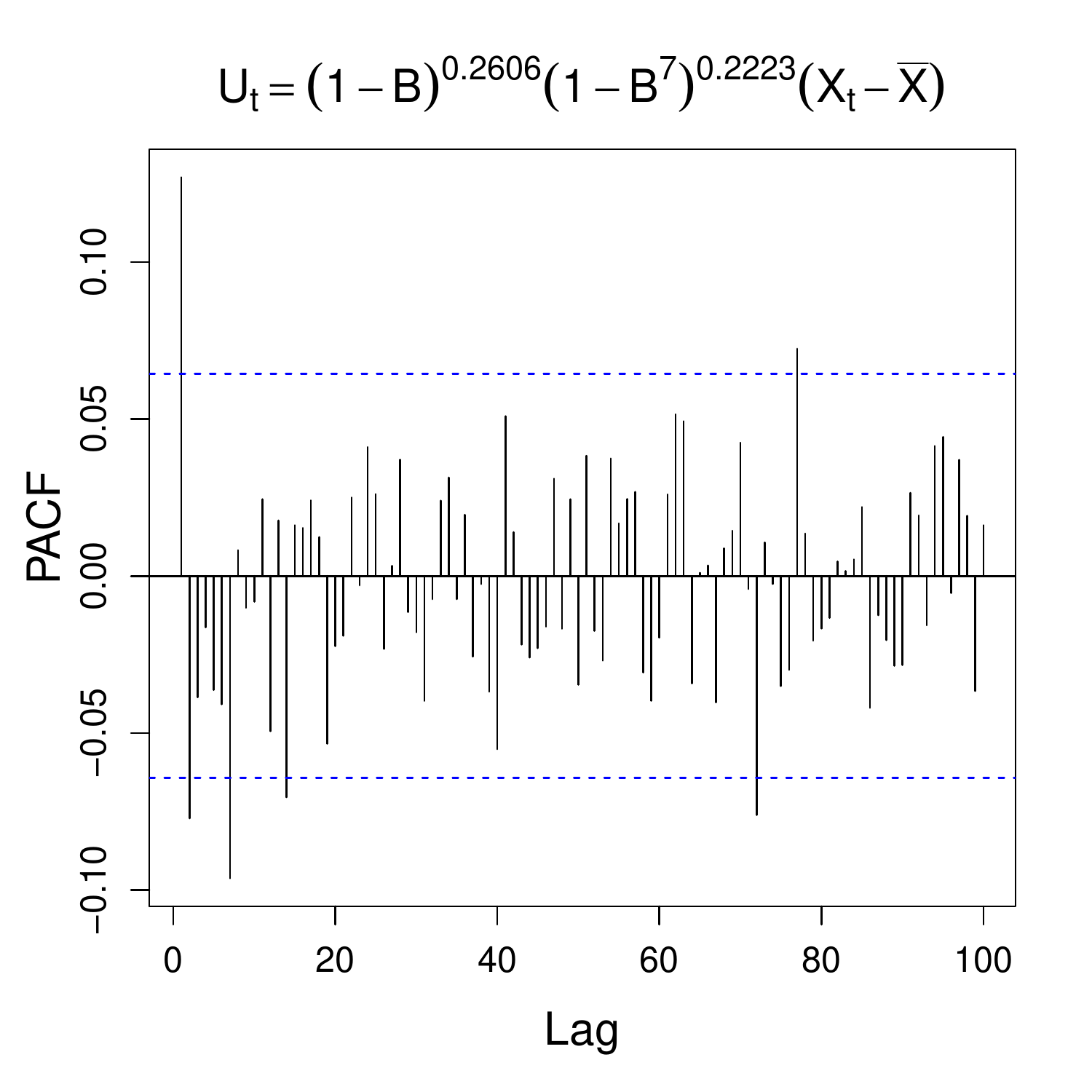}\label{fig_pacf_serie_filtrada}
}
\caption{ The ACF and PACF plots of $\hat{U}_t$.}
\end{figure}

Therefore, the model SARFIMA$(0,d,1)\times(0,D,1)_7$ was chosen for the $PM_{10}$ average data. The standard residual analysis did not present any anomaly of the residuals of this model, that is,  most of the correlations of $\hat{\epsilon}_t$ falls inside the confidence boundaries. Then, the residuals themselves appear to be uncorrelated. These are not presented here to save space but are available upon request. However, the plot in Figure \ref{fig_plot_residuo_quadrado} clearly indicates that the variance of the errors is not constant. Furthermore, the Figures \ref{fig_acf_residuo_quadrado} and \ref{fig_pacf_residuo_quadrado} are, respectively, the ACF and PACF of $\hat{\epsilon}^{2}_t$ and they suggest that a generalized conditional heteroscedasticity (GARCH) model can be suitable to capture the time-varying volatility in the data.

In order to statistically verify the presence of heteroscedasticity in $\hat{\epsilon}_t^2$, the Lagrange multipliers test was performed \citep{lee:1991}  and the null hypothesis of residual homecedasticity was rejected with $p-$value smaller than $0.001$. After performing model adequacy, the model GARCH(1,1) was adjusted for the $\hat{\epsilon}_t^2$ of the SARFIMA model. The final estimated model is a SARFIMA$(0,d,1)\times(0,D,1)_7$-GARCH$(1,1)$. The estimates of the parameters are displayed in Table \ref{tab2}.

\begin{figure}[!ht]
\centering
\subfigure[ Squared residuals (volatility) of PM$_{10}$ concentration]{
  \includegraphics[scale=0.55]{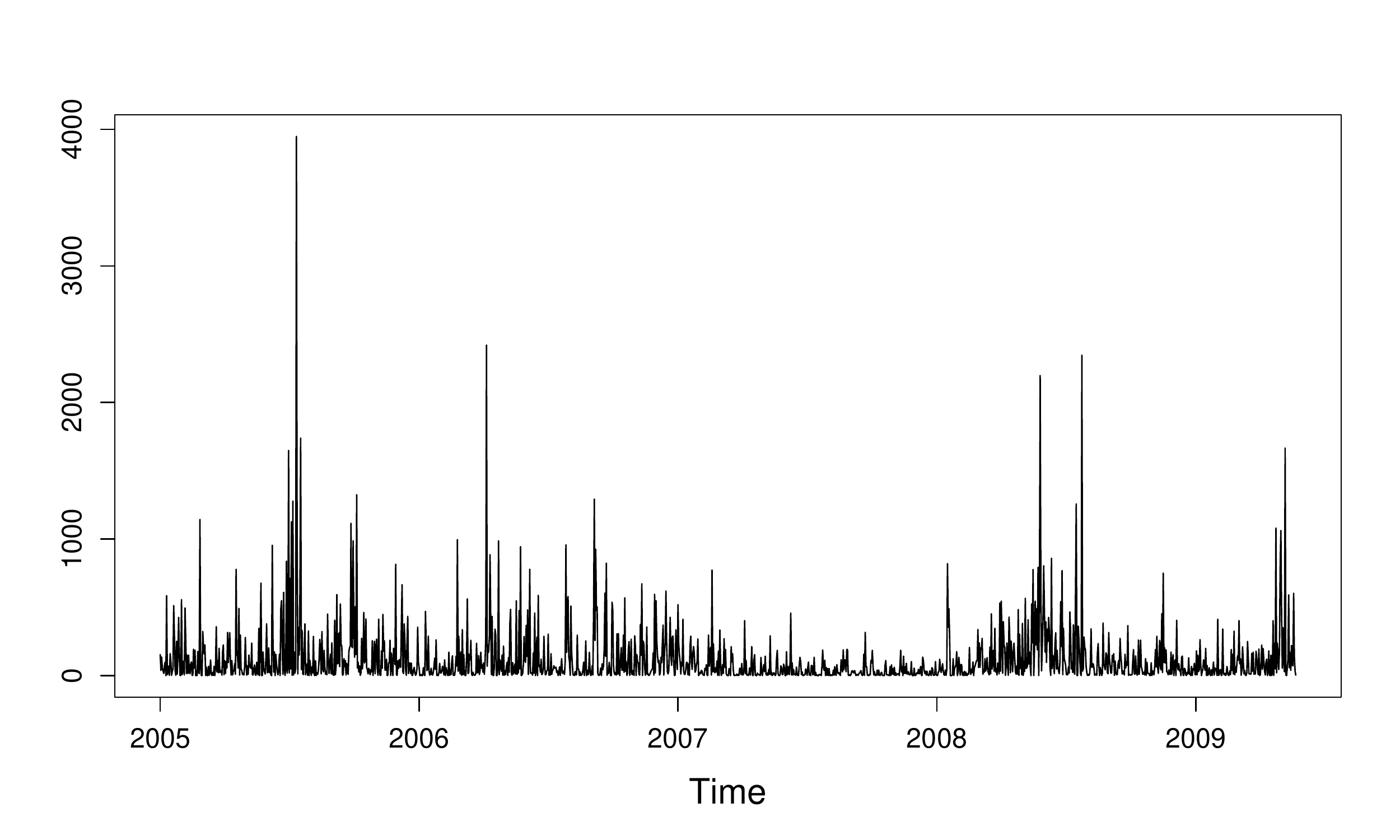}\label{fig_plot_residuo_quadrado}
}\\
\subfigure[ ACF ]{
  \includegraphics[scale=0.45]{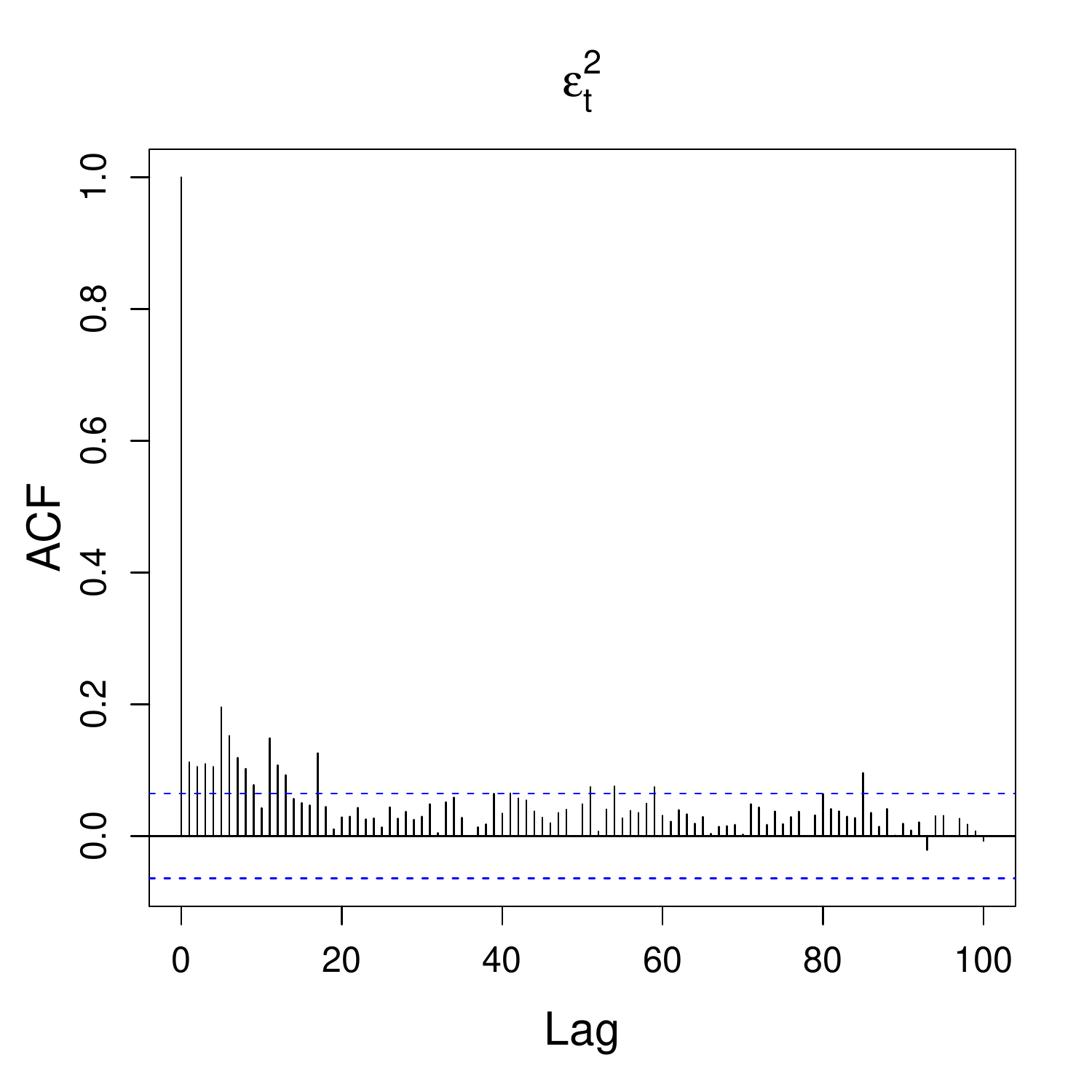}\label{fig_acf_residuo_quadrado}
}
\subfigure[ PACF ]{
  \includegraphics[scale=0.45]{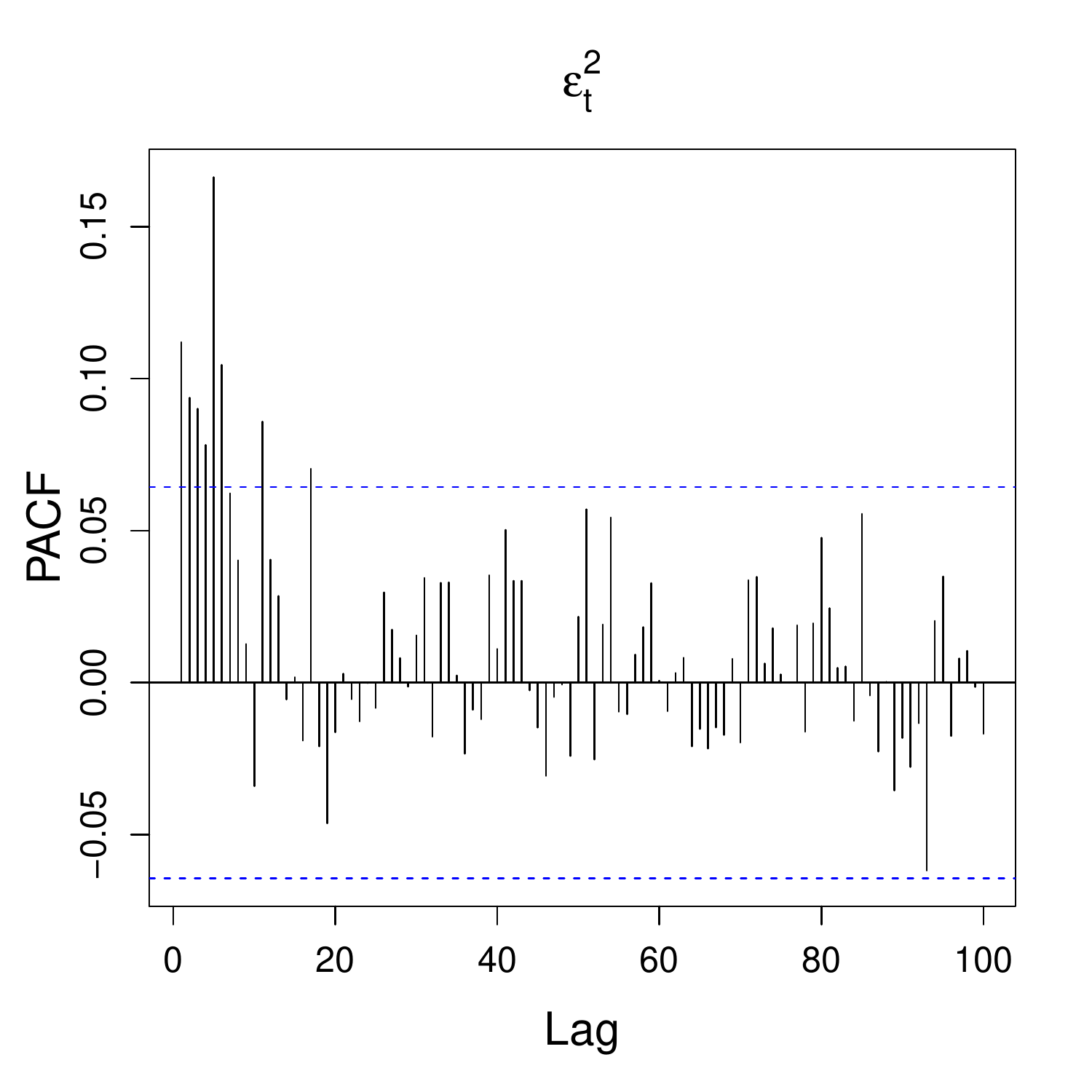}\label{fig_pacf_residuo_quadrado}
}
\caption{ Plots  related to the volatility of  PM$_{10}$ concentration}
\end{figure}

\begin{table}[!ht]
\centering
\caption{ SARFIMA-GARCH parameter estimates of PM$_{10}$ concentration}
  \begin{tabular}{ccccccccc}
  \hline\hline
  Parameter	&&	Estimate	&&	$s.d.$ &&	t-test	&& p-value\\
  \cline{1-1}\cline{3-3}\cline{5-5}\cline{7-7}\cline{9-9}
  $d$	&&	$0.2606$	&&	$0.0372$	&&	$7.0054$	&& $<0.0001$\\
  $D$	&&	$0.2223$	&&	$0.0596$	&&	$3.7299$	&& $0.0002$\\
  $\theta$	&&	$0.1417$	&&	$0.0258$	&&	$5.4923$	&& $<0.0001$\\
  $\Theta$	&&	$-0.1092$	&&	$0.0265$	&&	$-4.1208$	&& $<0.0001$\\
  $\alpha_0$	&&	$1.6464$	&&	$0.5623$	&&	$2.9280$	&& $0.0034$\\
  $\alpha_1$	&&	$0.0677$	&&	$0.0111$	&&	$6.0991$	&& $<0.0001$\\
  $\beta_1$	&&	$0.9205$	&&	$0.0132$	&&	$69.735$	&& $<0.0001$\\
  \hline\hline
  \end{tabular}\label{tab2}
\end{table}

The GARCH(1,1) model adequacy is now discussed. Figures \ref{hist_GARCH_residual} and \ref{acf_GARCH_res_sem_quadrado} present the histogram and the ACF of the residuals of the adjusted GARCH model. As a first analysis, these figures apparently indicate that the residuals are non correlated  and the histogram is  slightly positively skewed. A detailed investigation is as follows.  Statistical quantities of these residuals are given in  Tables \ref{tabd} and \ref{teste}. These confirm that the residuals are uncorrelated and not normally distributed, which was an expected result since the original data is also  right skewed.


\begin{figure}
\subfigure[Histogram]{
  \includegraphics[scale=0.45]{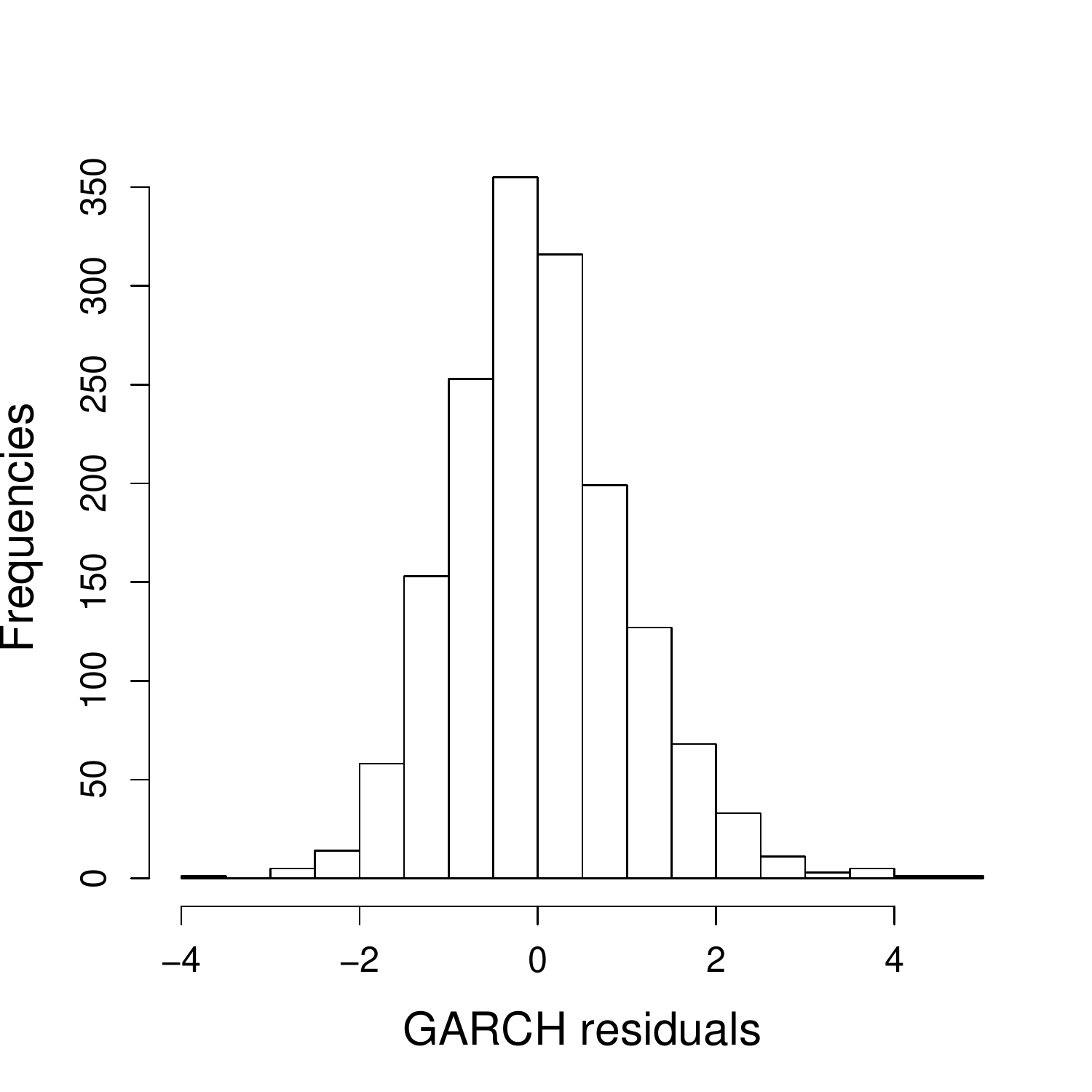}\label{hist_GARCH_residual}
}
\subfigure[The sample ACF]{
  \includegraphics[scale=0.45]{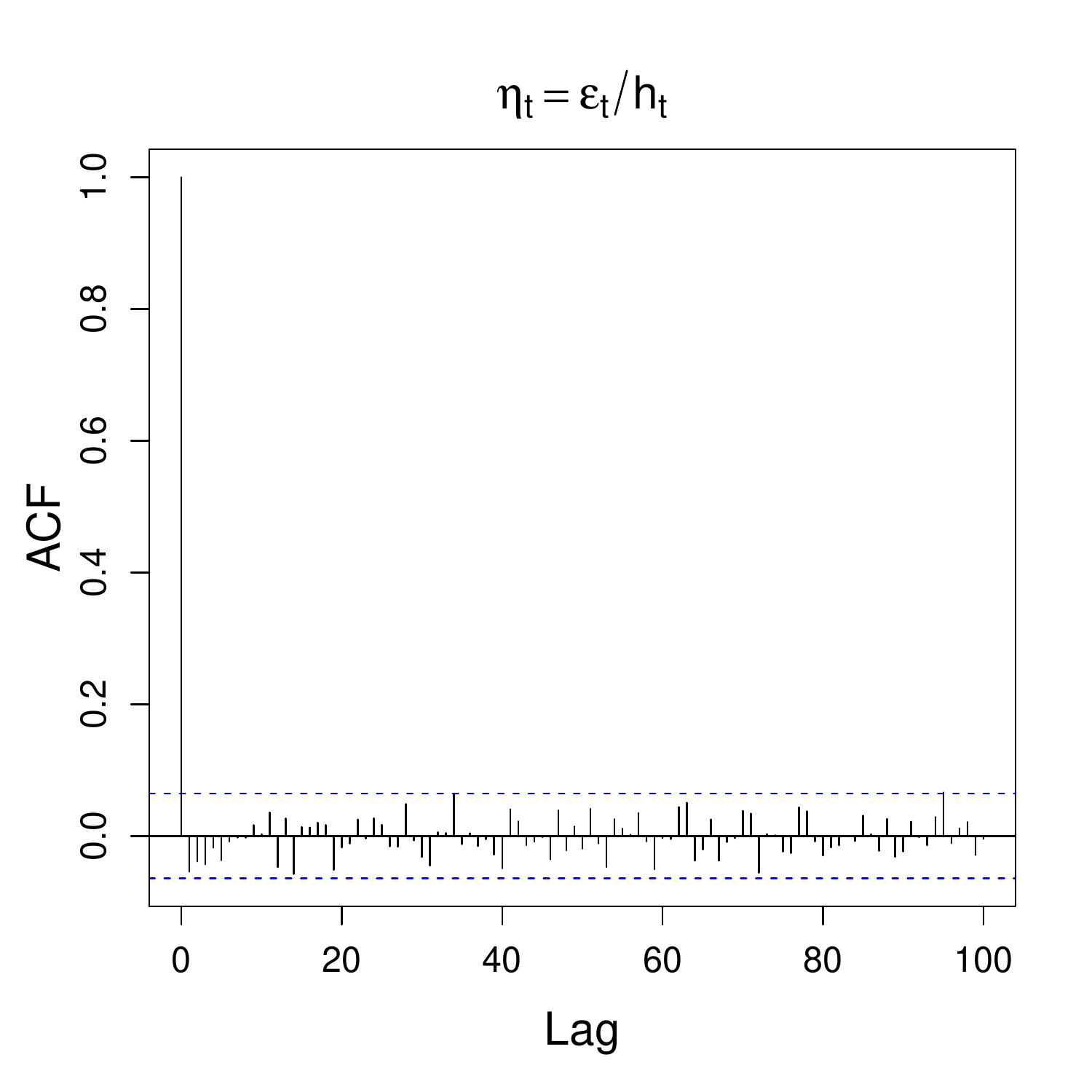}\label{acf_GARCH_res_sem_quadrado}
}
\caption{Plots of the residuals of the adjusted GARCH(1,1) model}
\end{figure}

\begin{table}[!ht]
\centering
\caption{Some statistics of the residuals of the adjusted volatility model} \begin{tabular}{ccccccc}
  \hline\hline
  Mean & & Stnd. dev. & & Skewness & & Kurtosis \\
  \cline{1-1}\cline{3-3}\cline{5-5}\cline{7-7}
  0.0128 & & 0.9994 & & 0.4277 & & 0.8718 \\
  \hline\hline
\end{tabular}\label{tabd}
\end{table}

\begin{table}[!ht]
\centering
\caption{Tests for normality \emph{($^*$)} and non correlation \emph{($^{**}$)}}
\begin{tabular}{ccccccc}
  \hline\hline
  Shapiro-Wilk$^*$ & & Jarque-Bera$^*$ & & Box-Pierce$^{**}$ & & Ljung-Box$^{**}$\\
  \cline{1-1}\cline{3-3}\cline{5-5}\cline{7-7}
  $<0.0001$ & & $<0.0001$ & & $0.1151^\dag$ & & $0.1138^\dag$ \\
  \hline\hline
\end{tabular}\label{teste}\vspace{-.2cm}
\begin{center}
$^\dag$These {\textit p}-values correspond to the Box-Pierce and Ljung-Box test statistics with lag 8.
\end{center}
\end{table}

To end the model adequacy, Figure \ref{ajuste} presents visual analysis of the SARFIMA adjusted model, that is, the one-step-ahead predicted values from year 2008 which indicates a reasonably good performance of the model here proposed. It can be seen that it was able to capture the tendency and seasonality of the series.

\begin{figure}[!ht]\centering
  \includegraphics[scale=0.55]{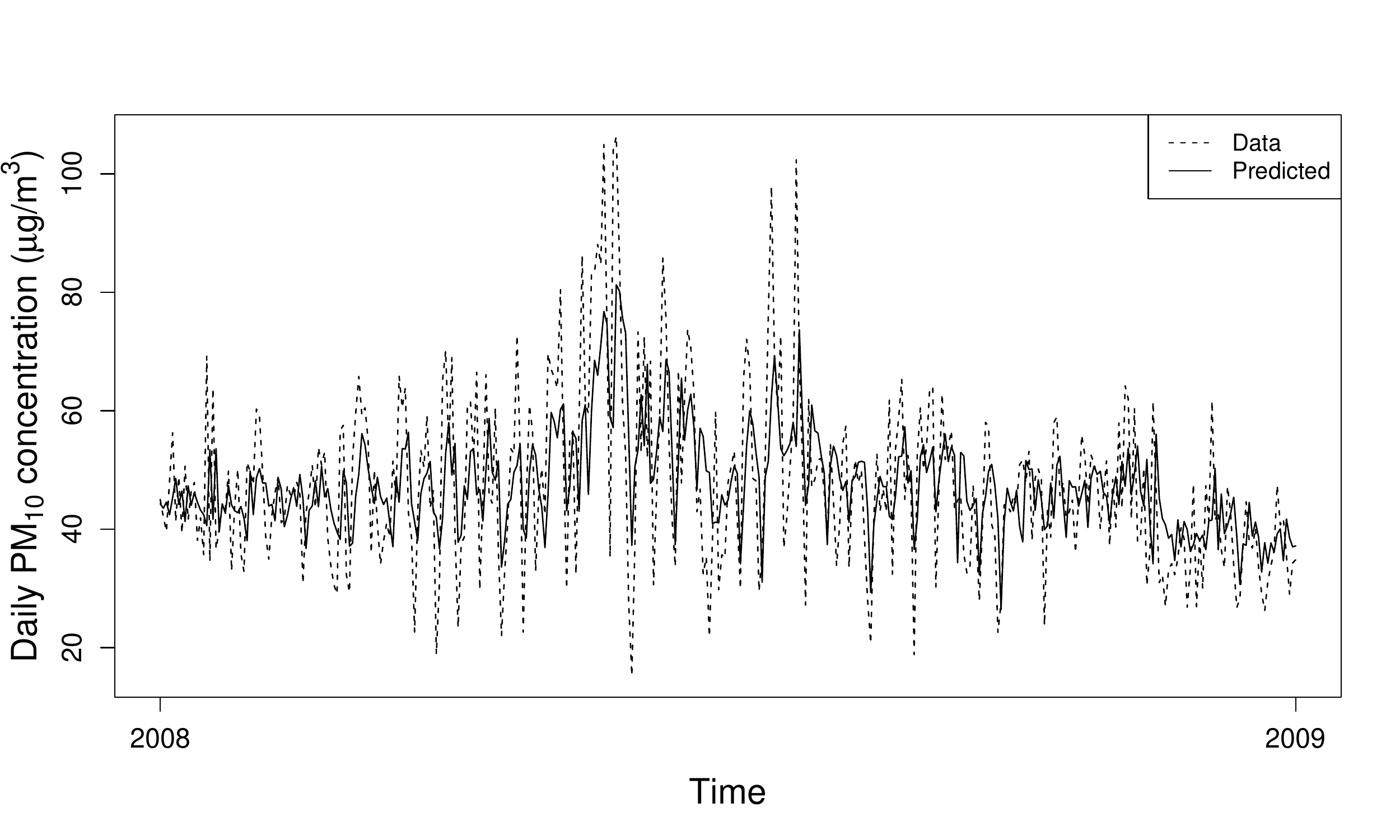}
\caption{PM$_{10}$ concentration and their predicted values from 01/01/2008 to 12/31/2009}\label{ajuste}
\end{figure}

\subsection{Forecasting issues}

This section examines the forecast performance of the model discussed in this paper with confidence intervals builded with homoscedastic and heteroscedastic variances.  As stated before, the observations from may 23th of 2009 to december 31th of 2009 were discarded from the modeling step (223 observations) to be used for an out-of-sample one-step-ahead forecast study. To measure the accuracy of the forecasts, the criterions used were the Mean Percentage Error (MPE) and the Mean Absolute Percentage Error (MAPE). To quantify the performance of the forecast intervals, the values of  the Coverage Percentage of  GARCH and  Homoscedastic Forecast Intervals, denoted by CPGFI and CPHFI, respectively, were calculated. These quantities are reported in Table \ref{tabperf}. The MPE and MAPE criterions indicated that the SARFIMA model here proposed gave reasonably accurate forecasts. Furthermore, the coverage percentage of the homoscedastic forecast interval CPFFI is much smaller than the confidence level of $95$\%. On the other hand, CPGFI is very close to the nominal confidence level, i.e., $CPGFI = 94.17\%$. This suggest that the SARFIMA-GARCH model well accommodates the properties of the  daily average PM$_{10}$ concentrations data set analyzed in this paper.

\begin{table}[!ht]
\centering
\caption{Forecast performance of the selected model}
\begin{tabular}{ccccccc}
 	\hline\hline
	\multicolumn{7}{c}{Criterions}\\
	\hline 	MPE && MAPE && CPGFI && CPHFI\\
	\cline{1-1}\cline{3-3}\cline{5-5}\cline{7-7}
	$8.46$\% && $23.85$\% && $94.17$\% && $91.03$\%\\
 	\hline\hline
\end{tabular}\label{tabperf}
\end{table}

Finally, Figure \ref{prev1} displays the observations and the out-of-sample one-step-ahead $95$\% GARCH and homoscedastic asymptotic forecast intervals for the model proposed. This figure provides a visual comparison of the coverage of these intervals. From this graph, one can see that the GARCH forecast intervals are able to capture the high volatility periods. This explain the coverage percentages showed in Table \ref{tabperf}.

\begin{figure}[!ht]
  \centering
    \includegraphics[scale=0.55]{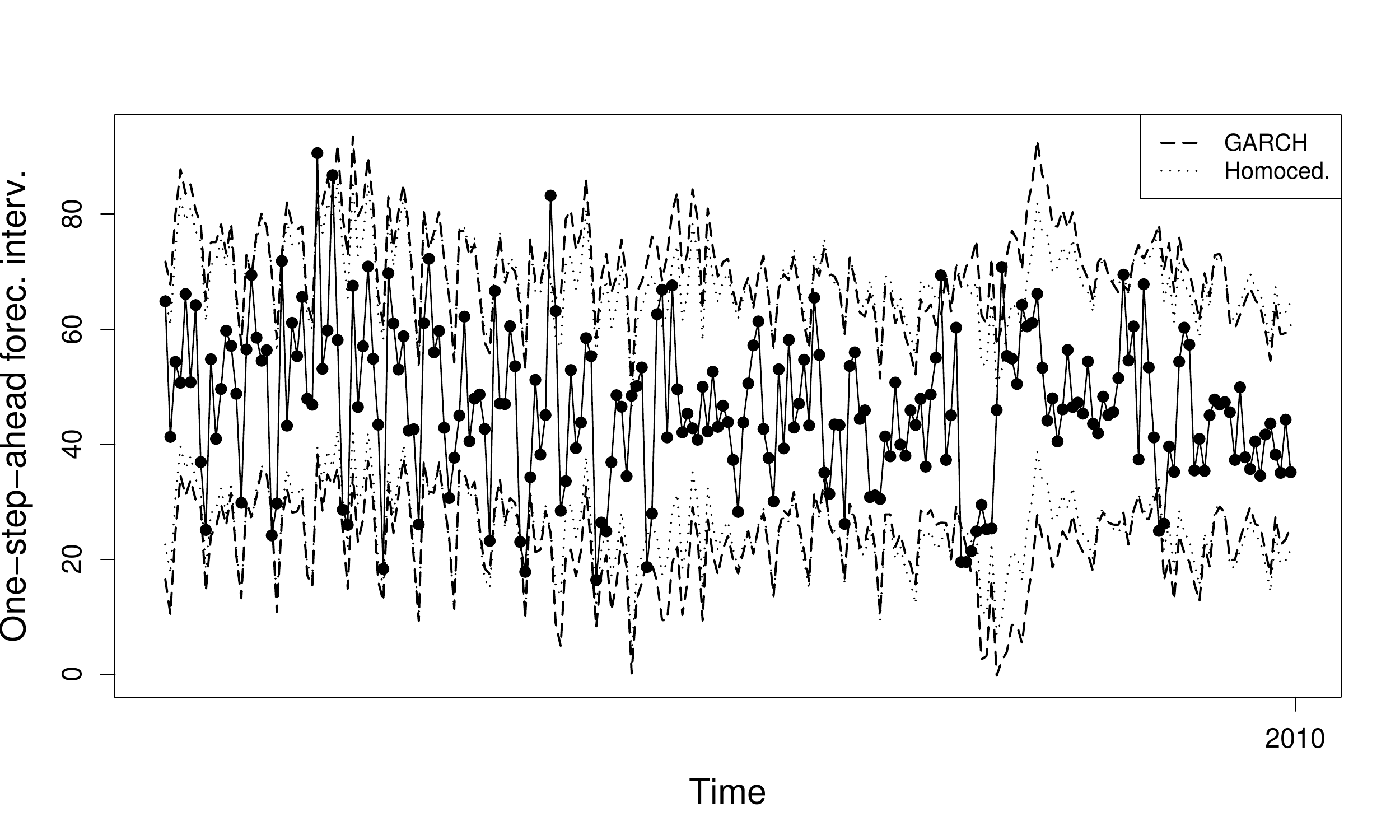}
    \caption{\small{GARCH and homocedastic 95\% forecasting intervals of the SARFIMA model of daily average PM$_{10}$ concentrations from 05/23/2009 to 12/31/2009}}
  \label{prev1}
\end{figure}

\section{Conclusions}

In this paper a seasonal ARFIMA model under heteroscedastic innovations is applied to model daily average PM$_{10}$ concentrations. To estimate the fractional parameters, the semiparametric procedure suggested in \citet{reisen:rodrigues:palma:2006a,reisen:rodrigues:palma:2006b} is considered under a non-constant conditional error variance. The memory estimates evidenced that the series is stationary with long-memory property at zero and seasonal frequencies. This is an interesting feature observed in the data which support the use of a more sophisticated model structure. Another equally interesting characteristic observed is that the conditional variance of the error is correlated. The features seasonality, long-memory and volatility of the data were well captured by the model proposed in this paper, that is, by the SARFIMA$(0,d,1)\times(0,D,1)_7$-GARCH$(1,1)$ model. The residual analysis and one-step ahead forecast indicated that the SARFIMA-GARCH model presented a very accurate model adequacy.

\section{Acknowledgements}
V. A. Reisen, N. Reis Jr and J. M. Santos gratefully acknowledge partial financial support from FAPES-ES, FACITEC-PMV-ES and CNPq/Brazil.

\bibliographystyle{ijf}
\bibliography{references_UTF8}

\section*{Appendix}
\begin{proof}[Proof of Proposition \ref{proparfimagarch}]
$(X_t)$ can be seen as as a fractional ARIMA process introduced by \citet{giraitis:leipus:1995} with garch-errors. Let $Y_t= \frac{\Phi(B^s)}{\Theta(B^s)}(1-B^s)^D X_t$. Then $Y_t$ follows an ARFIMA$(p,d,q)$-GARCH$(r,m)$ model according to \citet{ling:li:1997}. Under the assumptions the power expansions series $\frac{\Theta(z^s)}{\Phi(z^s)}(1-z^s)^{-D}$ and $\frac{\Phi(z^s)}{\Theta(z^s)}(1-z^s)^D$ converge for $|z|\leq1$. Then based on Theorems 2 and 2.3 in \citet{giraitis:leipus:1995} and \citet{ling:li:1997}, respectively, the statements (a) and (b) are straightforward obtained.

\end{proof}

\end{document}